\documentclass[11pt,documentclass,onecolumn]{article}
\usepackage{amsmath}
\usepackage{amsthm}
\usepackage{amssymb}
\usepackage{algorithm}
\usepackage{subcaption}
\usepackage{color}
\usepackage[english]{babel}
\usepackage{graphicx}
\usepackage{wrapfig,epsfig}
\usepackage{psfrag}
\usepackage{epstopdf}
\usepackage{url}
\usepackage{graphicx}
\usepackage{color}
\usepackage{epstopdf}
\usepackage{algpseudocode}
\usepackage{tikz}
\usepackage[hidelinks]{hyperref}
\usetikzlibrary{arrows}

\usepackage[margin=1in]{geometry}
\author{Sushrut Karmalkar\\\texttt{sushrutk@cs.utexas.edu}\\UT Austin \and Eric Price\\\texttt{ecprice@cs.utexas.edu}\\UT Austin}

\title{Compressed Sensing with Adversarial Sparse Noise via L1 Regression}

\newtheorem{theorem}{Theorem}[section]
\newtheorem{lemma}[theorem]{Lemma}
\newtheorem{definition}[theorem]{Definition}

\newtheorem{fact}[theorem]{Fact}
\newtheorem{claim}[theorem]{Claim}

\newcommand{\norm}[1]{\|#1\|}

\newcommand{\wh}{\widehat}

\newcommand{\eps}{\epsilon}

\newcommand{\R}{\mathbb{R}}

\newcommand{\RN}[1]{%
  \textup{\uppercase\expandafter{\romannumeral#1}}%
}

\newcommand{\vertiii}[1]{{\left\vert\kern-0.25ex\left\vert\kern-0.25ex\left\vert #1 
		\right\vert\kern-0.25ex\right\vert\kern-0.25ex\right\vert}}

\DeclareMathOperator*{\argmin}{arg\,min}
\DeclareMathOperator\erf{erf}

\usepackage{lineno}

\usepackage{bbm}
\allowdisplaybreaks
\usepackage{etoolbox}

\makeatletter

\newcommand{\define}[4][ignore]{%
  \ifstrequal{#1}{ignore}{}{
  \@namedef{thmtitle@#2}{#1}}%
  \@namedef{thm@#2}{#4}%
  \@namedef{thmtypen@#2}{lemma}%
  \newtheorem{thmtype@#2}[theorem]{#3}%
  \newtheorem*{thmtypealt@#2}{#3~\ref{#2}}%
}

\newcommand{\state}[1]{%
  \@namedef{curthm}{#1}
  \@ifundefined{thmtitle@#1}{
  \begin{thmtype@#1}
    }{
  \begin{thmtype@#1}[\@nameuse{thmtitle@#1}]
  }
    \label{#1}
    \@nameuse{thm@#1}
  \end{thmtype@#1}
  \@ifundefined{thmdone@#1}{
  \@namedef{thmdone@#1}{stated}%
  }{}
}

\newcommand{\restate}[1]{%
  \@namedef{curthm}{#1}
  \@ifundefined{thmtitle@#1}{
    \begin{thmtypealt@#1}
    }{
  \begin{thmtypealt@#1}[\@nameuse{thmtitle@#1}]
  }
    \@nameuse{thm@#1}
  \end{thmtypealt@#1}
  \@ifundefined{thmdone@#1}{
  \@namedef{thmdone@#1}{stated}%
  }{}
}

\newcommand{\thmlabel}[1]{
  \@ifundefined{thmdone@\@nameuse{curthm}}{\label{#1}
    }{\tag*{\eqref{#1}}}
}
\makeatother

\begin{document}
	
	\maketitle

\begin{abstract}
	We present a simple and effective algorithm for the problem of
	\emph{sparse robust linear regression}.  In this problem, one would like to estimate a sparse vector $w^* \in \R^n$ from linear
	measurements corrupted by sparse noise that can arbitrarily change an adversarially chosen $\eta$ fraction of measured responses $y$, as well as introduce bounded norm noise to the responses.
	
	For Gaussian measurements, we show that a simple algorithm based on
	L1 regression can successfully estimate $w^*$ for any
	$\eta < \eta_0 \approx 0.239$, and that this threshold is tight for
	the algorithm.  The number of measurements required by the algorithm
	is $O(k \log \frac{n}{k})$ for $k$-sparse estimation, which is
	within constant factors of the number needed without any sparse
	noise.

	Of the three properties we show---the ability to estimate sparse, as
	well as dense, $w^*$; the tolerance of a large constant fraction of
	outliers; and tolerance of adversarial rather than distributional
	(e.g., Gaussian) dense noise---to the best of our knowledge, no
	previous result achieved more than two.
\end{abstract}

\section{Introduction}


Linear regression is the problem of estimating a signal vector from
noisy linear measurements.  It is a classic problem with applications
in almost every field of science.  In recent decades, it has also
become popular to impose a sparsity constraint on the signal vector.
This is known as ``sparse recovery'' or ``compressed sensing'', and
(when the assumption holds) can lead to significant savings in the
number of measurements required for accurate estimation.

A well-known problem with the most standard approaches to linear
regression and compressed sensing is that they are not robust to
outliers in the data.  If even a single data point $(x_i, y_i)$ is
perturbed arbitrarily, the estimates given by the algorithms can also
be perturbed arbitrarily far.  Addressing this for linear regression
is one of the primary focuses of the field of robust
statistics~\cite{huber}.  Unfortunately, while the problem is clear,
the solution is not---no fully satisfactory robust algorithms exist,
particularly for high-dimensional data.

In this paper, we consider the model of robustness in which only the
responses $y_i$, not the features $x_i$, are corrupted by outliers.
In this model, if the features $x_i$ are i.i.d.\ normal, we show that
the classic algorithm of L1 minimization performs well and has fairly
high robustness, for both dense and sparse linear regression.  In
particular, we consider the observation model
\begin{align}\label{eq:signalmodel}
y = Xw^* + \zeta + d
\end{align}
where $X \in \R^{m \times n}$ is the observation matrix,
$w^* \in \R^n$ is the $k$-sparse signal, $\zeta \in \R^m$ is an
$\eta m$-sparse noise vector, and $d \in \R^m$ is a (possibly dense)
noise vector.  We will focus on the case of $X$ having i.i.d.\
$N(0, 1)$ entries, but the core lemmas and techniques can apply
somewhat more generally.

Without adversarial corruptions---i.e.\ if $\eta = 0$ so
$\zeta = 0$---this would be the compressed sensing problem. The most
standard solution for compressed sensing~\cite{CRT06} is L1
minimization: if $m > \Theta(k \log\frac{n}{k})$ then with high probability
\[
\wh{w} := \argmin_{\norm{y - Xw}_2 \leq \sigma} \norm{w}_1
\]
for any $\sigma > \norm{d}_2$ will satisfy
$\norm{\wh{w} - w^*}_2 \leq O(\sigma/\sqrt{m})$.  Unfortunately, this algorithm
is not robust to sparse noise of large magnitude: a single faulty
measurement $y_i$ can make the $(\norm{y - Xw}_2 \leq \sigma)$
ball infeasible.

To make the algorithm robust to sparse measurement noise, a natural approach is to replace the (non-robust) $\ell_2$ norm with the (robust) $\ell_1$ norm, as well as to swap the objective and the constraint. This ensures that the constrained parameter does not involve outliers. In this paper we show that this approach works, i.e., we show that
\begin{align}\label{eq:l1min}
\wh{w} := \argmin_{\norm{w}_1 \leq \lambda} \norm{y - Xw}_1
\end{align}
is a robust estimator for $w^*$.  In the following theorem, we show that~\eqref{eq:l1min} is robust to any fraction of corruptions $\eta$ less than
$\eta_0 := 2\left( 1 - \Phi(\sqrt{2 \log 2})\right) \approx 0.239$,
where $\Phi: \R \to [0,1]$ is the standard normal CDF.  If
$\lambda = \norm{w^*}_1$, the reconstruction error is
$O(\norm{d}_1/m)$; for larger $\lambda$, it additionally grows with
$\lambda - \norm{w^*}_1$:

\define[Sparse Case]{thm:main}{Theorem}{Let $\eta < \eta_0 - \eps$ where $\eps > 0$, and let
	$X \in \R^{m \times n}$ have i.i.d.\ $N(0, 1)$ entries with
	$m > C \frac{\alpha^2}{\epsilon^2} k \log(\frac{en}{\alpha^2 \eps k})$ for some large enough constant $C$ and
	parameter $\alpha \geq \frac{2}{\epsilon}$.  Then with probability $1-e^{-\Omega(\epsilon^2m)}$
	the matrix $X$ will have the following property: for any
	$y = Xw^* + d + \zeta$ with $\norm{w^*}_0 \leq k$ and
	$\norm{\zeta}_0 \leq \eta m$,
	\[  \wh{w} := \argmin_{\norm{w}_1 \leq \lambda} \norm{y - Xw}_1 \]
	for $\lambda \geq \norm{w^*}_1$ satisfies
	\[
	\norm{w^* - \wh{w}}_2 \leq O\left( \frac{1}{\epsilon - \frac{1}{\alpha}} \left(\frac{1}{m}\norm{d}_1 \right)  +  \frac{\lambda - \norm{w^*}_1}{\alpha \sqrt{k}} \right).
	\]}

\define[Dense Case]{thm:main2}{Theorem}{	 Let $\eta < \eta_0 - \eps$ where $\eps > 0$, and let
	$X \in \R^{m \times n}$ have i.i.d.\ $N(0, 1)$ entries with
	$m > C\frac{n}{\eps^2}$ for some large enough constant $C$.  Then with probability $1-e^{-\Omega(\eps^2 m)}$
	the matrix $X$ will have the following property: for any
	$y = Xw^* + d + \zeta$ with 
	$\norm{\zeta}_0 \leq \eta m$, \[\wh{w} := \argmin_{w} \norm{y - Xw}_1 \] satisfies
	\[ \|\widehat{w} - w^*\|_2 \leq O\left(\frac{\|d\|_1}{\epsilon m}\right)\] }

\state{thm:main}

In the case where $w^*$ is not sparse, the reconstruction error is shown to be $O(\norm{d}_1/m)$ in $O(n)$ samples using essentially the same proof. 
\state{thm:main2}

\paragraph*{Robustness threshold $\eta_0$.}  We show in
Section~\ref{sec:lower} that the threshold $\eta_0$ in
Theorem~\ref{thm:main} is tight for the algorithm: for any
$\eta > \eta_0$, there exist problem instances where the algorithm
given by~\eqref{eq:l1min} is not robust.  It remains an open question
whether any polynomial time algorithm can be robust for all
$\eta < 0.5$.

\subsection{Proof outline}

Our main result follows from a simple analysis of the fact that for well-behaved matrices $X$, $\ell_1$ regression recovers from adversarial corruptions. In this section we consider the illustrative case where there is no dense noise, in the limit of infinitely many samples. Let $(X_g, y_g)$ and $(X_b, y_b)$ denote the submatrices of $(X, y)$ corresponding to the uncorrupted and corrupted samples respectively and let $\widehat{w}$ denote the solution of $\ell_1$ regression. By definition $\widehat{w}$ satisfies \[ \|X\widehat{w} - y\|_1 \leq \|X w^* - y\|_1.\] Partitioning these $1$-norms into terms corresponding to good and bad samples, we get
\begin{align*}
0 &\geq  \|X \widehat{w} - y\|_1 - \|X w^* - y\|_1\\
&=  \left(\| X_{g} \widehat{w} - y_g \|_1 - \| X_{g} w^* - y_g \|_1\right) + \left(\| X_{b} \widehat{w} - y_b \|_1 - \| X_b w^* - y_b \|_1 \right)
\end{align*}
Observe that since we have no dense noise, $X_g w^* = y_g$.  An application of the triangle inequality then results in
\begin{align*}
0 &\geq  \| X_g (\widehat{w} - w^*)\|_1 + \left(\| X_{b} \widehat{w} - y_b \|_1 - \| X_b w^* - y_b \|_1 \right)\\
&\geq  \| X_g (\widehat{w} - w^*)\|_1  -  \| X_b (\widehat{w} - w^*)\|_1 
\end{align*}
i.e. 
\begin{equation}\label{eq:no_dense_noise}
0 \geq \| X_g (\widehat{w} - w^*)\|_1 -  \|X_b(\widehat{w}-w^*)\|_1 
\end{equation}
We now show that as long as $\eta < \eta_0 - \epsilon$ for constant $\epsilon > 0$, the right hand side above is $\geq C(\epsilon)\|\widehat{w} - w^*\|_2$ for some $C(\epsilon) > 0$ whenever $X$ is a Gaussian matrix. This will force $\widehat{w} = w^*$.

Equation~\eqref{eq:no_dense_noise} is minimized when the adversary corrupts the entries with the largest value for $|\langle x_i, w^* - \wh{w}\rangle|$. 
For any vector $v$, observe that in the limit of infinitely many samples, the histogram of the entries of $Xv$ is the same as that of $N(0, \|v\|_2^2)$. Let $t$ be chosen such that $\frac{1}{\sqrt{2\pi}} \int_{t}^{\infty}  e^{-\frac{x^2}{2}} dx  = \frac{\eta}{2}$. This makes ~\eqref{eq:no_dense_noise} proportional to 

\begin{equation}\label{eq:proportional}
\|w^* - \widehat{w}\|_2 \left( \int_{0}^{t} x e^{-\frac{x^2}{2}} dx - \int_{t}^{\infty} x e^{-\frac{x^2}{2}} dx \right). 
\end{equation}

At $\eta = \eta_0$ the $\ell_1$ norms of the largest $\eta$ and $1-\eta$ fraction of samples drawn from a Gaussian distribution are equal, and hence~\eqref{eq:proportional} is $0$. If $\eta < \eta_0 - \epsilon$ for constant $\epsilon$, this difference is proportional to the standard deviation, i.e. $\|\widehat{w} - w^*\|_2$. Our proof proceeds by showing that a minor variant of this argument works even in the presence of dense noise, and in $O(k \log \frac{n}{k})$ samples the empirical $\ell_1$ norms involved in the proof are close to the $\ell_1$ norms of the Gaussian distribution.  


\subsection{Related Work}

\paragraph*{Classical robust statistics.}  The classical robust
statistics literature on regression (see \cite{huber}) has developed a
number of estimators with breakdown point 0.5 (i.e., that are robust
for any $\eta < 0.5$).  However, all such known estimators need time
exponential in the data dimension $n$; the results also typically do
not deal with sparsity in $w^*$ and have distributional assumptions on
the dense noise $d$.  On the other hand, the results in this
literature usually also protect against corruption in $X_i$, not just
$y_i$; the L1 estimator is not robust to such corruptions.

\paragraph*{Recent progress in robust statistics.}  
There has been a lot of progress in the last year in the field of
robust statistics, leading to polynomial time algorithms with positive
breakdown points that are robust to corruptions in both $X$ and
$y$~\cite{2018arXiv1Diakonikolas,SEVER,liu2018high,klivans2018efficient}.
However, these results all focus on the performance for small $\eta$
(often required to be less than a non-explicit constant), do not
consider sparse $w^*$, and have additional restrictions on the dense
noise (typically that it be i.i.d.\ Gaussian,
although~\cite{klivans2018efficient} is somewhat more general).  In
Section~\ref{sec:empirical} we empirically compare L1 regression to
the algorithm of~\cite{2018arXiv1Diakonikolas} for the dimension 1
case, and find that the algorithm seems to have the same breakdown
point $\eta_0$ as L1 minimization under corruptions to $y$.

\paragraph*{L1 minimization in statistics.}  Known as L1 minimization
or Least Absolute Deviation, the idea of minimizing $\norm{y - Xw}_1$
actually predates minimizing $\norm{y - Xw}_2$, originating in the
18th century with Boscovich and Laplace~\cite{LAD1st}.  It is widely known
to be more robust to outliers in the $y_i$. However the extent to
which this holds depends on the distribution of $X$.   Surprisingly, we
have not been able to find a rigorous analysis of L1 minimization
for Gaussian $X$ that simultaneously achieves these three features of our
analysis: (1) an estimate of the breakdown point $\eta_0$ under
corruptions to the $y_i$; (2) an extension of the algorithm to sparse
$w^*$; or (3) a tolerance for adversarial $d$, rather than with a
distributional assumption.

L1 minimization is typically dismissed in the statistics literature as
being ``inefficient'' in the sense that, if the noise $d$ is i.i.d.\
Gaussian, L1 minimization requires about 56\% more samples than least
squares~\cite{yu2017robust} to achieve the same accuracy.  However
from the typical perspective of theoretical computer science, in which
constant factors are less important than the avoidance of
distributional assumptions, we find that L1 minimization is a very
competitive algorithm.

\paragraph*{L1 minimization in compressed sensing.}  Our error bound of
$\norm{d}_1/m$ is always better than the traditional
$\norm{d}_2 / \sqrt{m}$ bound for compressed sensing. The two bounds
match up to constant factors if the noise has a consistent magnitude,
but our bound is significantly better if the noise is heavy tailed.
The fact that our bound can drop the top $\eta$ fraction of noise
elements makes the distinction even more pronounced.

\paragraph*{Robust regression in the presence of label corruptions.}
The past few years have featured a number of polynomial time
algorithms for the problem considered in this paper, of sparse
regression in the presence of sparse corruptions to the labels. As is
typical in compressed sensing, there are approaches based on convex
programming and on iterative methods.

One natural algorithm for the problem is to try to learn both the
(sparse) signal and (sparse) noise, treating this as a single
compressed sensing problem with the bigger ``measurement matrix'' of
$X$ atop a (scaled) identity matrix.  With scaling $1/\lambda$, the
standard L1 minimization approach to compressed sensing is equivalent
to the following algorithm: minimize $\|w\|_1 + \lambda \|\zeta\|_1$
subject to $\| Xw + \zeta - y \|_2 \leq \epsilon$.  If the adjoined
measurement matrix satisfies an RIP-like property, then $w$ (and
$\zeta$) will both be recovered.

Such an approach was first introduced in \cite{laska2009exact} with
$\lambda = 1$, giving an algorithm that could tolerate up to about
$\eta \approx 1/(\log n)$ fraction sparse corruptions. This was then
improved by \cite{CSLi2011} by setting
$\lambda = \frac{1}{\sqrt{\log(en/m)}}$, improving the breakdown
point $\eta$ to an unspecified constant; naively following the proof
would give a value below 1\%.  We suspect that this approach -- which
recovers $\zeta$ as well as $w$ -- does not have a breakdown point
close to $\eta_0$.

The second class of algorithms for the problem are based off iterative
hard thresholding, where in each iteration one ignores the samples
that make a large error with the $\ell_2$ minimizer.  In
\cite{bhatia2015robust} it was shown that without any dense noise,
this yields exact recovery with a breakdown point of $\eta <
0.015$. \cite{consistent_robust} provided an algorithm that can handle
dense Gaussian noise, but the perturbations are required to be
oblivious to the matrix $X$ and the breakdown point is $0.0001$.

Another line of work, including \cite{Robust_LASSO,novel,nguyen2013exact}, considers non-adversarial
corruption.  For example, if the corruptions are in random locations,
and the signs of the signal vector are random, then one can tolerate
corruption of nearly 100\% of the $y_i$~\cite{nguyen2013exact}.
Finally,~\cite{WLJ07} considers the (essentially equivalent) LASSO
version of our proposed algorithm~\eqref{eq:l1min}, and shows that it
is robust to i.i.d.\ heavy-tailed median-zero noise. 

Thus, for sparse regression with both adversarial corruption of the
labels and dense noise, no previous polynomial-time algorithm
had a breakdown point above $0.015$.  We improve that to $0.239$ with a simple algorithm.

\paragraph*{LP Decoding and Privacy}

Very closely related to our work is that of~\cite{dwork2007price},
which gets very similar results to our dense-case results
(Theorem~\ref{thm:main2}) in the service of a privacy application.
This work observes the same threshold $\eta_0$ as we do for the same
L1-regression algorithm, but with a somewhat weaker error guarantee
(requiring a bound on $\norm{d}_\infty$ not $\norm{d}_1$).
\cite{dwork2007price} also proves that if $X$ is i.i.d.\ $\pm 1$
rather than Gaussian, the breakdown point would be positive but
strictly below $\eta_0$.  The subsequent work~\cite{WXT10} also
observes that $\ell_p$ regression for $p < 1$ would yield greater
breakdown points than $\eta_0$ for Gaussian $X$, similar to our Section~\ref{sec:Lp}.

\section{Definitions and notation}\label{sec:notation}

We start by defining a notion of robustness that we will use later. A matrix $X$ is said to be $(\eta, q)$-robust if for any submatrix consisting of an $\eta$ fraction of the rows, the $\ell_q$ norm of the submatrix times a unit vector is upper bounded by a constant times $m^{1/q}$. Also, for any submatrix consisting of a $1-\eta$ fraction of the rows, the $\ell_q$ norm of the submatrix times a unit vector is \emph{lower bounded} by a constant times $m^{1/q}$. 
\begin{definition}
	A matrix $X \in \mathbb{R}^{m \times n}$ is said to be \emph{$(\eta,q)$-robust} with respect to $U \subset \mathbb{R}^n$ if there exist constants $S_{U,\eta}^{\max}$ and $S_{U, \eta}^{\min}$ satisfying the following conditions for all $v \in U$.
	\[\max_{\substack{{S \subset [m]}\\ |S| \leq \eta m}} \|(X v)_S\|^q_q \leq m \cdot S_{U, \eta}^{\max} \cdot \|v\|_2 \] 
	\[\min_{\substack{{S \subset [m]}\\|S| \geq (1-\eta) m} } \|(X v)_S\|^q_q \geq m \cdot S_{U,\eta}^{\min} \cdot \|v\|_2. \]
\end{definition}
We now define some notation. $S^{*}_{k, *}$ and $S^{*}_{\eta}$ will be used to refer to the robustness constants with respect to $k$-sparse vectors and $\mathbb{R}^n$ respectively. $v_T$ will denote the vector $v$ with all entries whose indices are outside $T$ set to $0$. 

We use $\Phi$ to denote the CDF of  $N(0,1)$. $B(\gamma)$ and $G(\gamma) $ will be used to refer to the $\ell_1$ norm of the largest (in absolute value) $\gamma$ fraction and the smallest $1-\gamma$ fraction with respect to the Gaussian distribution respectively, i.e. 
%


\[ B(\gamma) = \frac{2}{\sqrt{2 \pi}} \int_{\Phi^{-1}(1-\frac{\gamma}{2})}^{\infty} ze^{-\frac{z^2}{2}} dz = \sqrt{\frac{2}{\pi}}\left(e^{-(\erf^{-1}(1-\gamma))^2}\right) \]
and
\[ G(\gamma) = \frac{2}{\sqrt{2 \pi}} \int_{0}^{\Phi^{-1}(1-\frac{\gamma}{2})} ze^{-\frac{z^2}{2}} dz = \sqrt{\frac{2}{\pi}}\left(1-e^{-(\erf^{-1}(1-\gamma))^2}\right). \]

Define $\eta_0$ to be the largest $\eta$ such that $G(\eta) \geq B(\eta)$. Using the expressions above, one can solve for $\eta$ to get $\eta_0 = 2(1-\Phi(\sqrt{2\log2})) \approx 0.239$. 

Also, we will use $f(x) \lesssim g(x)$ to mean there are constants $X$ and $C$ such that $\forall x > X. |f(x)| \leq C |g(x)|$. 

\section{Robustness of Gaussian matrices} 

In the following lemma, we show that Gaussian matrices are $(\eta, 1)$-robust with constants in terms of $B(\cdot), G(\cdot)$ defined earlier. 

\begin{lemma}\label{lem:samp_comp}
	Let $X$ be an $m \times n$ Gaussian matrix, where $m \geq \frac{C}{\epsilon^2} \cdot \left( k \log \frac{en}{k\epsilon} + \log \frac{1}{\delta} \right)$ for a large enough constant $C$ and $\epsilon < 1$.  Then with probability $1-\delta$, $X$ is $(\eta, 1)$ robust with constants
	\[ S_{k, \eta}^{\min} = G(\eta - \epsilon) - \epsilon \]
	\[ S_{k,\eta}^{\max} = B(\eta + \epsilon) + \epsilon. \]
\end{lemma}
\begin{proof}
	Rearranging terms in the definition we see that we would like to show 
	\[\frac{1}{m} \max_{\substack{{S \subset [m]}\\ |S| \leq \eta m}} \left\|\left(X \cdot \frac{v}{\|v\|}\right)_S\right\|_1 \leq B(\eta + \epsilon) + \epsilon  \] 
	\[\frac{1}{m} \min_{\substack{{S \subset [m]}\\|S| \geq (1-\eta) m} } \left\|\left(X \cdot \frac{v}{\|v\|}\right)_S\right\|_1 \geq G(\eta - \epsilon) - \epsilon. \]
	Without loss of generality it is sufficient to prove the above for all ($k$-sparse) unit vectors. Let $x_i$ denote the $i^{th}$ row of $X$ and let $S_v = \{ \langle x_i, v \rangle \mid i \in m\}$. Note that $S_v$ look like samples from $N(0,1)$ for any fixed unit vector $v$. Before we continue, we define some notation. Let $ \widehat{G}_v(\eta)$ denote the  smallest possible $\ell_1$ norm of a subset of $S_v$ of size $(1-\eta)m$ and let $\widehat{B}_v(\eta)$ be defined similarly to denote the largest possible $\ell_1$ norm of any subset of size $\eta m$. What we want to prove is
	\[ \widehat{B}_v(\eta) < B(\eta + \epsilon) + \epsilon\]
	and
	\[ \widehat{G}_v(\eta) > G(\eta - \epsilon) - \epsilon \]
	for all $k$-sparse unit vectors $v$. To do this, we will first prove that the relationship holds with high probability for all $k$-sparse unit vectors in a fine enough net on the sphere, and then say that the deviation cannot be very large for points outside the net. 
	
	We will need the following fact proven in Appendix \ref{appendix:A}. Here $\widehat{G}(\eta)$ refers to the $\ell_1$ norm of the smallest $(1-\eta)$ fraction of $S$ with respect to the uniform distribution and $\widehat{B}(\eta)$ refers to the $\ell_1$ norm of the largest $\eta$ fraction of $S$ with respect to the uniform distribution. 
	\begin{fact}\label{ell1_est}
		Let $S = \{ z_1, \dots, z_m \}$ be i.i.d. samples from  $N(0,1)$.Then with probability $1-O\left(e^{\frac{m\epsilon^2}{2}}\right)$,
		\[ \widehat{G}(\eta) > G(\eta - \epsilon) - \epsilon\]
		and
		\[ \widehat{B}(\eta) < B(\eta+\epsilon) + \epsilon \]
	\end{fact}
	
	For now, let $v$ be a fixed vector and let $\tau > 0$ be a parameter. Define the following bad events  
	\[ \mathcal{G}_{v} =\left \{\widehat{G}_v(\eta) < G(\eta - \epsilon) - \epsilon \right \}\] 
	\[\mathcal{B}_v = \left \{ \widehat{B}_v(1-\eta) > B(1-\eta + \epsilon) + \epsilon \right\}\]
	\[\mathcal{N} =  \left \{ \forall i \in [m], \|x_i\|_2 > \sqrt{n + \tau} \right \}.\]
	These events correspond to either the $\ell_1$ norms of the smallest $1-\eta$ fraction or the largest $\eta$ fraction not being close enough to the expectation, or the $2$-norm of the Gaussian vectors not being close enough to the expectation. Applications of Fact \ref{ell1_est} for $\eta$ and $1-\eta$, and concentration for $\chi^2$ random variables then implies
	\[ \Pr\left(\mathcal{G}_v \lor \mathcal{B}_v \right) \lesssim e^{\frac{m\epsilon^2}{2}}\]
	and
	\[ \Pr\left( \mathcal{N} \right) \lesssim  me^{-\frac{n\tau^2}{8}}.\]
	For a unit $\ell_2$ ball in a $k$-dimensional subspace of $\mathbb{R}^n$, there exists a $\gamma$-net of size $(1+\frac{2}{\gamma})^k < (\frac{3}{\gamma})^k$. Let $C$ be the union of these nets over all subspaces corresponding to $k$-sparse vectors. A union bound now gives us 
	\[ \Pr \left( \exists v \in C :\mathcal{G}_v \lor \mathcal{B}_v \right) \lesssim  \binom{n}{k} \left(\frac{3}{\gamma}\right)^k \left(e^{\frac{m\epsilon^2}{2}}\right)  \]
	
	We will now move from the net to the union of all $k$-sparse unit balls. Let $u \in \mathbb{R}^n$ be a $k$-sparse unit vector. Then for any $t$, there exist $v_0, \dots, v_t \in C$ having the same support as $u$ and a unit vector $d$ also having the same support as $u$, such that \[u = \sum_{i=0}^{t} \gamma^i v_i + \gamma^{t+1} d.\] This follows from choosing $v_0$ to be the closest point in the net to $u$, choosing $v_1$ to be the closest point in the net to $(u - v_0)/\gamma$ and so on. 
	
	Let $U \subset [m]$ be the set of indices of $X$ corresponding to the smallest (in absolute value) $(1-\eta)$ fraction of elements of  $S_u$.  Conditioning on the bad events not happening (i.e. on the event $\overline{\left( \exists v \in C. \mathcal{G}_v \lor \mathcal{B}_v \right) \lor\mathcal{N} }$) we see
	\begin{align*}
	\widehat{G}_u(\eta) &= \frac{1}{m}\sum_{i \in U} \left| \left \langle x_i, \sum_{j=0}^{t} \gamma^j v_j + \gamma^{t+1} d \right\rangle \right|\\
	&\geq \frac{1}{m} \sum_{i \in U} |\langle x_i, v_0\rangle | - \frac{1}{m} \sum_{i \in U} \left|\left\langle x_i, \sum_{j=1}^t \gamma^j v_j + \gamma^{t+1} d\right\rangle\right|\\
	&\geq \frac{1}{m} \sum_{i \in U} |\langle x_i, v_0\rangle | -  \sum_{j=1}^t \left(\frac{\gamma^j }{m} \sum_{i \in U} |\langle x_i,v_j\rangle| \right)  - \frac{\gamma^{t+1}}{m} \sum_{i \in U} |\langle x_i, d\rangle|\\
	&\geq \widehat{G}_{v_0}(\eta) - \sum_{j=1}^{t} \widehat{B}_{v_j}(1-\eta)\gamma^j- \frac{\gamma^{t+1}}{m} \sum_{i \in U}\|x_i\|\|d\|\\
	&\geq (G(\eta - \epsilon) - \epsilon) -  (B(1-\eta +\epsilon) + \epsilon) \sum_{j=1}^t \gamma^j - \gamma^{t+1} \sqrt{n+\tau}\\
	&\geq G(\eta - \epsilon) - \epsilon - 2\gamma (B(1-\eta +\epsilon) + \epsilon)  - \gamma^{t+1} \sqrt{n+\tau}\\
	&\geq  G(\eta - \epsilon) - \epsilon - 4\gamma - \gamma^{t+1} \sqrt{n+\tau}\\
	&\geq G(\eta - \epsilon) - 2\epsilon
	\end{align*}
	The first few inequalities are a consequence of the definitions of $G_v$ and $B_v$ and the third inequality follows from an application of Cauchy-Schwartz. The second to last inequality follows by noting that $\epsilon < 1$ and $B(1-\eta +\epsilon)  < 1$, and the final inequality follows by setting $t > \log \frac{n+\tau}{\epsilon}$ and $\gamma = \frac{\epsilon}{10}$. This means
	\begin{align*} 
	&\Pr \left(\text{There exists a $k$-sparse unit $v$ such that } G(\eta - \epsilon) - \widehat{G}_v(\eta) > 2\epsilon \right)\\ &\lesssim  \binom{n}{k} \left(\frac{30}{\epsilon}\right)^k \left( e^{-\frac{m\epsilon^2}{2}}\right) + me^{-\frac{n \tau^2}{8}}\\
	&\lesssim e^{k \log \frac{en}{k} + k \log \left(\frac{30}{\epsilon}\right) -\frac{m\epsilon^2}{2}} + m e^{-\frac{n \tau^2}{8}}
	\end{align*}
	Setting $\tau = 10mn\log \frac{1}{\delta}$ and $m \gtrsim \frac{1}{\epsilon^2} \cdot \left( k \log \frac{en}{k\epsilon}  + \log \frac{1}{\delta} \right)$ makes the bound on the probability above $\lesssim \delta$. The result now follows by rescaling $\epsilon$ and $\delta$ appropriately. 
\end{proof} 

The previous lemma showed that the Gaussian matrix is robust with respect to truly $k$-sparse vectors. However, we will need to show that it is robust with respect to $(w^* - \widehat{w})$, i.e. the difference between the true vector and the solution of $\ell_1$ regression. To do this, we will use a standard shelling argument to transfer upper and lower bounds for the restricted eigenvalues over $(1+\alpha^2)k$-sparse vectors to the restricted eigenvalues over the cone $V_S = \{ v \in \mathbb{R}^n \mid \Delta +\|v_S\|_1 \geq \|v_{\overline{S}}\|_1 \}$ for some $S$ satisfying $|S| = k$, which is the cone in which this difference lies. This is the content of the following lemma from Appendix~\ref{appendix:B}. 

\define[Shelling Argument]{lem:shelling}{Lemma}{ Let $A \in \mathbb{R}^{m \times n}$ satisfy 
	\[ L \|v\|_2 \leq \|A v\|_1 \leq U\|v\|_2\]
	for all $(1+\alpha^2)k$-sparse vectors $v$. If $S \subset [m]$ is fixed and of cardinality $k$, then $A$ satisfies 
	\[  \frac{L}{1+\alpha} \left( \alpha  -  \frac{U}{L} \right)\|v\|_2 -  \frac{2U\Delta}{ \alpha \sqrt{k}}  \leq \|Av\|_1 \leq U \left(1 + \frac{1}{\alpha}\right) \|v\|_2 +  \frac{U \Delta}{\alpha \sqrt{k}} \]
	for all \[v \in V_S = \{ v \in \mathbb{R}^n \mid \Delta +\|v_S\|_1 \geq \|v_{\overline{S}}\|_1 \}. \] 
}
\state{lem:shelling}

We can now prove the main lemma which will be used to say that Gaussian matrices are robust with respect to the vector $w^* - \wh{w}$. 

\begin{lemma}[Main Lemma]\label{lem:gauss_robust}
	Let $\eta < \eta_0$,  $\epsilon \in (0,1)$, $\alpha > 1$ and $\Delta > 0$ be free parameters, and let $S \subset [n]$ be a fixed subset of size $k$. Let $X \in \mathbb{R}^{m \times n}$ be a matrix with entries drawn from $N(0,1)$ and suppose $m >  \frac{C}{\epsilon^2} \cdot \left( k\alpha^2 \log \frac{en}{k\alpha^2\epsilon}  + \log \frac{1}{\delta} \right)$  for some large enough constant $C$. Then with probability $1-\delta$, for all \[v \in V_S = \{ v \in \mathbb{R}^n \mid \Delta +\|v_S\|_1 \geq \|v_{\overline{S}}\|_1 \}\] and for all $T \subset [m]$ such that $|T| \leq \eta m$,
	\[  \|(Xv)_{\overline{T}} \|_1 - \|(Xv)_T\|_1 \gtrsim m\|v\|_2  \left( \left(G(\eta - \epsilon) - B(\eta + \epsilon)  - 2\epsilon\right) -  \frac{1}{\alpha} \right) - \frac{m\Delta}{\alpha\sqrt{k}}. \]

	%
\end{lemma}
\begin{proof}
	
	The matrix $X$ is both $(\eta,1)$ and $(1-\eta, 1)$ robust for all $k(1+\alpha^2)$-sparse vectors. An application of Lemma 3.3 for any submatrix $A$ of $X$ consisting of an $\eta$ fraction of it's rows gives us
	\[  \|Av\|_1 \leq mS^{\max}_{k(1+\alpha^2), \eta}\left(1+\frac{1}{\alpha} \right)\|v\|_2 + \frac{mS^{\max}_{k(1+\alpha^2), \eta} \Delta}{\alpha \sqrt{k}}.\]
	This proves 
	\[ \max_{\substack{{S \subset [m]}\\ |S| \leq \eta m}} \|(Xv)_S\|_1 \leq mS^{\max}_{k(1+\alpha^2), \eta}\left(1+\frac{1}{\alpha} \right)\|v\|_2 + \frac{mS^{\max}_{k(1+\alpha^2), \eta} \Delta}{\alpha \sqrt{k}} \]
	A similar application proves that for any matrix $A$ consisting of a $(1-\eta)$ fraction of the rows of $X$, we get 
	\[ \|A v\|_1 \geq \frac{mS^{\min}_{k(1+\alpha^2), \eta}}{1+\alpha}\left(\alpha-\frac{S_{k(1+\alpha^2), 1-\eta}^{\max} }{S^{\min}_{k(1+\alpha^2), \eta}} \right)\|v\|_2 -  \frac{2mS^{\max}_{k(1+\alpha^2), \eta}\Delta}{\alpha \sqrt{k}} \]
	i.e. 
	\[ \min_{\substack{{S \subset [m]}\\|S| \geq (1-\eta) m} } \|(Xv)_{\overline{S}}\|_1 \geq  \frac{mS^{\min}_{k(1+\alpha^2), \eta}}{1+\alpha}\left(\alpha-\frac{S_{k(1+\alpha^2), 1-\eta}^{\max} }{S^{\min}_{k(1+\alpha^2), \eta}} \right)\|v\|_2 -  \frac{2mS^{\max}_{k(1+\alpha^2), \eta}\Delta}{\alpha \sqrt{k}} \]
	To complete the proof, we now estimate the parameters involved. If $\eta < \eta_0$ and $\epsilon < 1$ and the underlying matrix is an $m \times n$ Gaussian matrix where $m \geq \frac{C}{\epsilon^2} \cdot \left( k\alpha^2 \log \frac{en}{k\alpha^2\epsilon}  + \log \frac{1}{\delta} \right)$, Lemma \ref{lem:samp_comp} yields 
	\[S_{k(1+\alpha^2), 1-\eta}^{\max} < B(1-\eta+\epsilon) + \epsilon < 1 + 1 = 2.\]
	Similar applications of Lemma \ref{lem:samp_comp} give us that $S_{k(1+\alpha^2), \eta}^{\max}$ and $\frac{S_{k(1+\alpha^2), 1-\eta}^{\max} }{S^{\min}_{k(1+\alpha^2), \eta}}$ are also upper bounded by constants. This results in the following bounds
	\[  \max_{\substack{{S \subset [m]}\\ |S| \leq \eta m}} \|(Xv)_S\|_1 \leq mS^{\max}_{k(1+\alpha^2), \eta}\|v\|_2\left(1+\frac{1}{\alpha} \right) +  \frac{2m\Delta}{\alpha\sqrt{k}},\]
	\[ \min_{\substack{{S \subset [m]}\\|S| \geq (1-\eta) m} } \|(Xv)_{S}\|_1 \geq  mS^{\min}_{k(1+\alpha^2), \eta}\|v\|_2 \left( \frac{\alpha - C}{\alpha+ 1}\right) - \frac{4m\Delta}{\alpha\sqrt{k}}.\]
	By taking the difference of the above inequalities and simplyfing, we get the following for any $T \subset [m]$ such that $|T| \leq \eta m$ 
	\begin{align*} 
	&\|(Xv)_{\overline{T}} \|_1 - \|(Xv)_T\|_1 \\&\geq m\|v\|_2 \left( S^{\min}_{k(1+\alpha^2), \eta} \left( \frac{\alpha - C}{\alpha+ 1}\right)
	- S^{\max}_{k(1+\alpha^2), \eta}\left(1+\frac{1}{\alpha} \right)	\right) -   \frac{6m\Delta}{\alpha \sqrt{k}} \\
	&\geq m\|v\|_2 \left( S^{\min}_{k(1+\alpha^2), \eta} - S^{\max}_{k(1+\alpha^2), \eta} - \left(  S^{\min}_{k(1+\alpha^2), \eta} + S^{\max}_{k(1+\alpha^2), \eta} \right) \frac{C+1}{\alpha} \right) -   \frac{6m\Delta}{\alpha \sqrt{k}}\\ 
	&\geq m\|v\|_2 \left( \left(G(\eta - \epsilon) - B(\eta + \epsilon)  - 2\epsilon\right) - \frac{4(C+1)}{\alpha} \right) -  \frac{6m\Delta}{\alpha \sqrt{k}}\\
	&\geq m\|v\|_2 \left( \left(G(\eta - \epsilon) - B(\eta + \epsilon)  - 2\epsilon\right) - \frac{1}{\alpha} \right) -  \frac{m\Delta}{\alpha \sqrt{k}}\\
	\end{align*}
	
\end{proof}


\section{Proof of main theorem}

\restate{thm:main}
\begin{proof}
	Let $X_g$ and $X_b$ denote $X$ restricted to the rows that are not corrupted, and to the rows that are corrupted respectively. Let $y_g$ and $y_b$ denote the corresponding $y$ terms. By the definition of $\widehat{w}$ and noting that $w^*$ is feasible for the program,
	\begin{align*}
	0 &\geq  \|X \widehat{w} - y\|_1 - \|X w^* - y\|_1\\
	&=  \left(\| X_{g} \widehat{w} - y_g \|_1 - \| X_{g} w^* - y_g \|_1\right) + \left(\| X_{b} \widehat{w} - y_b \|_1 - \| X_b w^* - y_b \|_1 \right) \\
	&\geq \| X_g (\widehat{w} - w^*)\|_1 - 2 \| X_g w^* - y_g \|_1  + \left(\| X_{b} \widehat{w} - y_b \|_1 - \| X_b w^* - y_b \|_1 \right)\\
	&\geq \| X_g (\widehat{w} - w^*)\|_1 - 2 \| X_g w^* - y_g \|_1 -  \|X_b(\widehat{w}-w^*)\|_1
	\end{align*}
	Where the second equality follows from  $\| X v - y \|_1 = \| X_{g} v - y_g \|_1 + \| X_{b} v - y_b \|_1$, and the inequalities are just applications of the triangle inequality. Rearranging terms now gives us
	\begin{equation} \label{eqn:main}
	2 \| X_g w^* - y_g \|_1 \geq \| X_g (\widehat{w} - w^*)\|_1 -  \|X_b(\widehat{w}-w^*)\|_1
	\end{equation}
	Let $\widehat{w}= w^* + h$ and let $S$ be the support of $w^*$. Then
	\begin{align*}
	\lambda &\geq \|\widehat{w}\|_1  \\
	&= \|h + w^*\|_1 \\
	&\geq \|w^*\|_1 + \|h_{\overline{S}}\|_1 - \|h_S\|_1 \\
	\implies (\lambda - \|w^*\|_1) + \|h_S\|_1 &\geq  \|h_{\overline{S}}\|_1 
	\end{align*}
	Setting $\Delta = (\lambda - \|w^*\|_1)$ and $T$ to be the set of corrupted indices in Lemma \ref{lem:gauss_robust} implies that if $m \gtrsim  \frac{1}{\epsilon^2} \cdot \left( k\alpha^2 \log \frac{en}{k\alpha^2\epsilon}  + \log \frac{1}{\delta} \right)$, then with probability $1-\delta$
	\begin{align*}
	& \| X_g (\widehat{w} - w^*)\|_1 -  \|X_b(\widehat{w}-w^*)\|_1 \\
	&= \|(Xh)_{\overline{T}} \|_1 - \|(Xh)_T\|_1\\
	&\gtrsim m\|h\|_2 \left(G\left(\eta - \frac{\epsilon}{2}\right) - B\left(\eta + \frac{\epsilon}{2}\right)  -\epsilon - \frac{1}{\alpha} \right) - \frac{m(\lambda - \|w^*\|_1)}{\alpha \sqrt{k}} 
	\end{align*}  
	Combining this with~\eqref{eqn:main}, as long as the coefficient of $\|h\|_2$ is positive, we get
	\begin{equation}\label{eq:main2}
	\|\widehat{w} - w^*\|_2 \lesssim \frac{1}{\left(G\left(\eta - \frac{\epsilon}{2}\right) - B\left(\eta + \frac{\epsilon}{2}\right)  -\epsilon - \frac{1}{\alpha} \right)} \left( \frac{\|d\|_1}{m} + \frac{(\lambda - \|w^*\|_1)}{\alpha \sqrt{k}}\right)
	\end{equation}
	It turns out 
	\[\left(G\left(\eta - \frac{\epsilon}{2}\right) - B\left(\eta + \frac{\epsilon}{2}\right)  - \epsilon \right) \gtrsim \epsilon.\] 
	This follows by a simple lower bound via the Taylor expansion of $1 - 2e^{-(\erf^{-1}((1-\eta_0) + x))}$ around $x = 0$. 
	\begin{align*}
	G\left(\eta - \frac{\epsilon}{2}\right) - B\left(\eta + \frac{\epsilon}{2}\right)  - \epsilon &\geq \sqrt{\frac{2}{\pi}}\left(1 - 2e^{-(\erf^{-1}(1-\eta + \frac{\epsilon}{2}))^2} \right) - \epsilon \\
	&= \sqrt{\frac{2}{\pi}}\left(1 - 2e^{-(\erf^{-1}(1-\eta_0 + (\epsilon + \frac{\epsilon}{2})))^2} \right) - \epsilon \\
	&\geq 3 \sqrt{2 \log 2} \cdot \epsilon - \epsilon\\
	&\gtrsim \epsilon
	\end{align*}
	i.e.
	\begin{equation}\label{eq:G-B}
	G\left(\eta - \frac{\epsilon}{2}\right) - B\left(\eta + \frac{\epsilon}{2}\right)  - \epsilon - \frac{1}{\alpha} \gtrsim \eps- \frac{1}{\alpha}
	\end{equation}
	
	Substituting our terms back into \eqref{eq:main2} gives us,
	\[  O\left(\frac{1}{\epsilon - \frac{1}{\alpha}} \cdot \frac{\|d\|_1}{m} + \frac{\lambda -\|w^*\|_1}{\alpha \sqrt{k}} \right) \geq \|\widehat{w} - w^*\|_2. \]
\end{proof}
We also note that in the case that $w^*$ is not sparse, one can directly use Lemma \ref{lem:samp_comp} once we get to \eqref{eqn:main} and continue the proof from there. This results in the following theorem. 

\restate{thm:main2}
Note that if there is no dense noise (i.e. $d = 0$), the above theorem immediately gives exact recovery when the fraction of corruptions is  $\eta < \eta_0 - \eps$. 


\section{$\ell_p$ regression for $0 < p < 1$}\label{sec:Lp}
Define $\ell_p$ regression to be the problem of recovering a signal by minimizing the $p^{th}$ power of the $\ell_p$ norm, i.e. 
\[ \widehat{w} = \argmin_{v} \sum_{i=1}^m |\langle x_i, v \rangle - y_i|^{p}.\]
Observe that $0 < p < 1$ implies
\[ \left(\sum_i a_i\right)^p \leq \sum_i a_i^{p}. \] This allows a proof similar to that of Theorem \ref{thm:main} to go through. We make the following claim. 
\begin{claim}\label{clm:l_p_dense_adversarial}
	
	Let $X$ be an $(\eta, p)$-robust matrix where $p \in (0,1]$. Then for any $\eta < \alpha$ the solution of $\ell_p$ regression, $\widehat{w}$ satisfies
	
	\[ \frac{1}{(S^{\min}_{\eta} - S^{\max}_{\eta})}  \cdot \frac{\|d\|_p}{m} \gtrsim \|\widehat{w} - w^*\|_2^p, \] 
	where $\|d\|_p = \sum_{i=1}^m |d_i|^p$ and $\alpha$ is the threshold below which $(S^{\min}_{\alpha} - S^{\max}_{\alpha}) > 0$ begins to hold. 
\end{claim}
\begin{figure}
	\centering
	\includegraphics[width=0.48\textwidth]{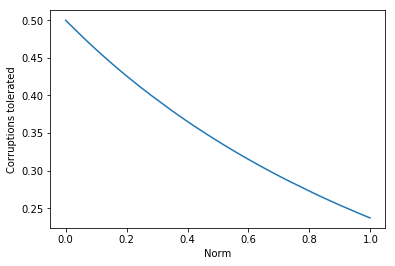}
	\captionsetup{aboveskip=-5pt}
	\caption{As the norm goes to $0$, in the limit of having infinite samples, $\ell_p$ regression can tolerate almost half the samples being corrupted.}
	\label{fig:norm_vs_corruptions}
\end{figure}
If $X$ is a Gaussian matrix, then as $p \rightarrow 0$, in the limit of a large number of samples, the value of $\eta$ at which the condition \[ (S^{\min}_{\eta} - S^{\max}_{\eta}) > 0\] begins to hold goes from $\eta_0$ to $0.5$. We plot the breakdown point against the norm in Figure \ref{fig:norm_vs_corruptions}. Unfortunately, $\ell_p$ regression in general seems to be NP-hard as well as approximation resistant.

\section{Lower bounds}\label{sec:lower}
In this section, we show that for the case of adversarial dense noise our results are tight for the $\ell_1$ regression algorithm. Recall our notation: $X$ is the matrix of $x_i$, $y = X w^* + \zeta + d$ where $\|\zeta\|_0 \leq \eta m$ and $d$ is the dense noise and $v_T$ denotes the vector $v$ with all entries with indices outside $T$ set to $0$.
\begin{theorem} 
	Let $m \gtrsim \frac{n}{\eps^2}$ and $0 < \eps < 0.2$, 
	\begin{itemize}
		\item[1.] If $\eta > \eta_0 + \eps$ and $d = 0$ (i.e. there is no dense noise), then there exists a choice for $\zeta$ such that $\ell_1$ regression does not exactly recover the original signal vector.
		\item[2.] Even if $\zeta = 0$ (i.e. there are no sparse corruptions),  there exists a choice for $d$ such that the solution of $\ell_1$ regression, $\widehat{w}$, satisfies \[ \|\widehat{w} - w^*\|_2 \gtrsim \frac{\|d\|_1}{m} \] 
	\end{itemize}
	
\end{theorem}
\begin{proof} 
	Let $T$ be the support of the largest $\eta m$  entries of $(Xw^*)$. For the first part, let $\zeta = -(Xw^*)_T$ and observe that since $d= 0$, the loss of the $0$ vector with respect to $y$ is $\|X_{\overline{T}} w^*\|_1$ and the loss of $w^*$ is $\|X_{T} w^*\|_1$. Since $m \gtrsim \frac{n}{\eps^2}$ we know that with probability $1-e^{-Cn}$ for some constant $C$,  
	\[\|X_T w^*\|_1 > \left(B \left(\eta - \frac{\eps}{2} \right) - \frac{\eps}{2}\right) \cdot m \] and \[\|X_{\overline{T}} w^*\|_1 < \left(G \left(\eta + \frac{\eps}{2}\right) + \frac{\eps}{2}\right) \cdot m.\] Hence, 
	\[ \|X_T w^*\|_1 - \|X_{\overline{T}} w^*\|_1 > \left(B \left(\eta - \frac{\eps}{2}\right) - G\left(\eta + \frac{\eps}{2}\right)- \eps\right) \cdot m \gtrsim  m \eps. \]
	The final inequality follows from a calculation similar to the one used to show~\eqref{eq:G-B}, by looking at the Taylor expansion of $B(\eta_0 + \frac{x}{2}) - G(\eta_0 + \frac{3x}{2}) - x$ around $x = 0$.  This implies $ \|X_T w^*\|_1  > \|X_{\overline{T}} w^*\|_1$ and so $\ell_1$ regression cannot return $w^*$ as the answer. 
	
	Let $T'$ be the support of the smallest $(1-(\eta_0 + \frac{\epsilon}{2}))m$ entries of $(Xw^*)$.
	For the second part, set $d = -(Xw^*)_{T'}$. Now, more than $(1-\eta_0)m$ entries of $y$ are $0$, and so $\ell_1$ regression will recover $0$. The resulting error in 2-norm is $\|\widehat{w} - w^*\|_2 = \|w^*\|_2$. Since $d = -(Xw^*)_{T'}$, $\|d\|_1$ is the $\ell_1$ norm of $Xw^*$ over the smallest $1-\eta_0 - \frac{\epsilon}{2}$ fraction of the indices. By arguments similar to earlier \[\|d\|_1 = \|(Xw^*)_{T'}\|_1 > m\left(G\left(\eta_0 - \frac{\eps}{2} \right) - \eps\right) \cdot \|w^*\|_2.\] It can be checked whenever $\eps < 0.2$, $G\left( \eta_0 -\frac{\eps}{2} \right) - \eps > 0.4$.  Hence, \[\|\widehat{w} - w^*\|_2 = \|w^*\|_2 \geq \frac{1}{G\left(\eta_0 -\frac{\eps}{2} \right) - \eps} \left( \frac{\|d\|_1}{m} \right) \gtrsim \frac{\|d\|_1}{m}.\]
\end{proof}



\section{Empirical comparisons to prior work}\label{sec:empirical}

\begin{figure}
	\centering
	\includegraphics[width=0.48\textwidth]{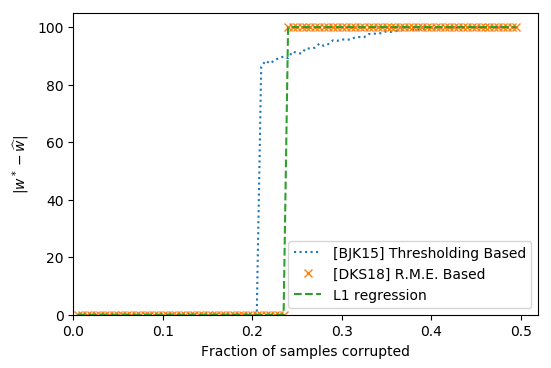}
	\captionsetup{aboveskip=5pt}
	\caption{Empirically in the one-dimensional case, the recovery
		threshold for $\ell_1$ regression and the robust mean
		estimation-based algorithm of \cite{2018arXiv1Diakonikolas}
		match at $\eta_0$.}
	\label{fig:error_v_corruption}
\end{figure}

We compare the tolerance of $\ell_1$ regression to algorithms from two recent papers
\cite{2018arXiv1Diakonikolas} and \cite{bhatia2015robust}. We study the fraction of corruptions these algorithms can tolerate in the limit of a large number of samples. Our experiment is the following - we study the one-dimensional case
where $w^* = 100$ and the adversarial noise is selected by setting the
largest $\eta$ fraction of observed $y$'s to $0$. We run the three algorithms on a dataset of $1000$ samples for $\eta$ ranging from $0$ to $0.5$ and consider the point when the algorithm stops providing exact recovery. In Figure \ref{fig:error_v_corruption} we plot the the error of the recovered $\widehat{w}$ from $w^*$ against the fraction of corruptions. 

While the fraction of corruptions tolerated by the algorithm from \cite{bhatia2015robust} for our example is more than what they prove in general (which is $\frac{1}{65}$), the fraction of corruptions it can tolerate is still less than that of $\ell_1$ regression on this example. For $\ell_1$ regression we observe what we have already proven earlier, that this example achieves our upper bound - i.e. it tolerates no more than an $\eta_0 \approx 0.239$ fraction of corruptions.

Curiously, the robust mean estimation based algorithm by \cite{2018arXiv1Diakonikolas} on this example tolerates exactly the same fraction of corruptions as $\ell_1$ regression.


%

\section*{Acknowledgements} 
The authors would like to thank the anonymous reviewers and Aravind Gollakota for helpful suggestions and comments about the writeup. 

\bibliographystyle{alpha}
\bibliography{main}

\section*{Appendix}
\appendix 

\section{Facts about N(0,1)}\label{appendix:A} 
In this section, let $\Phi$ and $\widehat{\Phi}$ denote the CDF of  $N(0,1)$ and the CDF of the uniform distribution over a set of samples drawn from $N(0,1)$ respectively -- the set will be clear from context. $B(\gamma)$ and $G(\gamma) $ refer to the $\ell_1$ norm of the largest (in absolute value) $\gamma$ fraction of the entries and the smallest $1-\gamma$ fraction of the entries with respect to the Gaussian distribution, and $\widehat{G}(\gamma)$ and $\widehat{B}(\gamma)$ are defined similarly but for the uniform distribution over samples from $N(0,1)$. 

\begin{fact}\label{fact:1}
	Let $S = \{ z_1, \dots, z_m \}$ be i.i.d. samples from  $N(0,1)$. Then for any $\tau, \gamma \in [0,1]$ the following holds with probability $1-4e^{-2m\tau^2}$.
	\[ \Phi^{-1}(\gamma-\tau) <  \widehat{\Phi}^{-1}(\gamma) <   \Phi^{-1}(\gamma +\tau). \]
\end{fact}
\begin{proof}
	The Dvoretzky-Kiefer-Wolfowitz inequality states 
	\begin{equation} \label{DKW}
	\Pr\Bigl(\sup_{x\in\mathbb R} \bigl(|\widehat{\Phi}(x) - \Phi(x)|\bigr) > \varepsilon \Bigr) \le 2e^{-2m\varepsilon^2}\qquad \text{for every }\varepsilon\geq\sqrt{\tfrac{1}{2m}\ln2}.
	\end{equation} 
	If $t = \Phi^{-1}(\eta)$, Equation \ref{DKW} then tells us that for any $\epsilon$ independent of $m$ (i.e. constant $\epsilon$), 
	\[  \Pr\Bigl(|\widehat{\Phi}(t) - \eta| > \varepsilon \Bigr) \le 2e^{-2m\varepsilon^2}\]
	i.e. 
	\[  \Pr\Bigl( \widehat{\Phi}(t) \leq \eta + \epsilon \Bigr) \le 2e^{-2m\varepsilon^2}.\]
	Setting $\eta = \gamma - \tau$ and $\epsilon = \tau$ we see 
	\[  \Pr\Bigl( \widehat{\Phi}(t) \leq \gamma \Bigr) \le 2e^{-2m\tau^2}.\]
	Monotonicity of $\widehat{\Phi}$ then proves the first inequality. The second inequality follows similarly.
\end{proof}

\begin{fact}
	Let $S = \{ z_1, \dots, z_m \}$ be i.i.d. samples from  $N(0,1)$ and let $\eta < \eta_0$. Then with probability $1-O\left(e^{\frac{m\epsilon^2}{2}}\right)$,
	
	\[ \widehat{G}(\eta) > G(\eta - \epsilon) - \epsilon\]
	and
	\[ \widehat{B}(\eta) < B(\eta + \epsilon) + \epsilon. \]
\end{fact}

\begin{proof}
	Consider the random variable $Y$ where $Y$ is $N(0,1)$ conditional from being drawn from $[-t,t]$ (i.e. $Y$ has the PDF of a truncated Gaussian distribution). If $\gamma < \frac{1}{2}$, for $t = \Phi^{-1}(1-\frac{\gamma}{2})$  
	\[ E[|Y|] = \frac{1}{1-\gamma} \int_{-t}^t |x| e^{-x^2/2} dx = \frac{1}{1-\gamma} G(\gamma). \]
	Observe that one can sample from $Y$ by sampling from $Z$ which is distributed as $N(0,1)$ and discarding samples outside $[-t,t]$. Since the PDF is scaled, we have to scale the empirical distribution as well
	\[ \widehat{E}[|Y|] = \frac{1}{m(1-\gamma)} \sum_{i=1}^m |z_i| \cdot 1_{[-t,t]} (z_i)   \] 
	Let $\gamma$ be such that $E|Y| \leq \frac{1}{2(1-\gamma)}$, then $|Y|$ has subgaussian tails with some constant parameter. To see this, observe that
	\begin{align*}
	E\left[e^{\lambda (|Y| - E[|Y|])}\right] &= \frac{2}{1-\gamma} \int_{0}^{t} e^{-x^2/2 + \lambda x - \lambda E[|Y|]} dx\\
	&\lesssim e^{ \frac{\lambda^2}{2} - (2 \cdot (1-\gamma))^{-1} \lambda}\\
	&\lesssim e^{O(\lambda^2/2 - \lambda)} 
	\end{align*}
	
	This implies concentration for the expectation \[ \Pr\left(\left|\widehat{E}[|Y|] - E[|Y|]\right| > \epsilon\right) < O\left(e^{-\frac{m \epsilon^2}{2} }\right).\] Multiplying both sides inside the probability by $(1-\gamma)$ and noting that since $\gamma < \frac{1}{2}$ this is bounded by $\frac{1}{2}$ we see 
	\[ \Pr\left(\left|\frac{1}{m} \sum_{i=1}^m |z_i| \cdot 1_{[-t,t]} (z_i) - G(\gamma)\right| > \frac{\epsilon}{2}\right) < O\left(e^{-\frac{m \epsilon^2}{2} }\right)\] 
	We now set $\gamma = \eta - \epsilon$. Since $\eta < \eta_0 \approx 0.239$ and $\epsilon> 0$, $\gamma < \frac{1}{2}$.  Fact \ref{fact:1} now implies that with probability  $1-2e^{-2m\epsilon^2}$, at most an $1-\eta$ fraction of the samples lie in $[-t,t]$. These have to be smaller in absolute value than the remaining samples. Since $\widehat{G}(\eta)$ is defined to be the $\ell_1$ norm of the $1-\eta$ fraction of points smallest in absolute value, we see
	\[ \widehat{G}(\eta) > \frac{1}{m} \sum_{i=1}^m |z_i| \cdot 1_{[-t,t]} (z_i).\] This implies \[ \Pr\left(\widehat{G}(\eta) > G(\eta  -\epsilon) - \epsilon \right) < O\left( e^{-2m\epsilon^2} +  e^{-\frac{m \epsilon^2}{2}}\right). \]
	The other direction is done similarly, however in this case $Y$ is the random variable gotten by conditioning samples from $N(0,1)$ to be outside $[-t,t]$. 
\end{proof}

\section{Shelling argument}\label{appendix:B}

\restate{lem:shelling}
\begin{proof} 
	The goal is to transfer bounds from the eigenvalues of $A$ restricted over the sparse vectors, to the eigenvalues of $A$ restricted over $V_S$. To this end we will select an element of $V_S$ and express it as a sum of sparse vectors. Applications of standard inequalities will then let us transfer bounds. 
	
	For any $v \in V_S$ partition $[n]$ into
	$S, T_1, \dots, T_{\frac{n-k}{\alpha^2 k}}$ where $T_i$ is the set of indices corresponding to the $i^{th}$ largest $\alpha^2 k$-sized set of elements from $v_{\overline {S}}$.
	
	We will now prove the upper and lower bounds on the eigenvalues for vectors restricted to the set $V_S$. The triangle inequality implies
	
	\[ \|Av_{S \cup T_1}\|_1 - \sum_{i > 1} \|A v_{T_i}\|_1 \leq \|Av\|_1 \leq \|Av_{S \cup T_1}\|_1 + \sum_{i > 1} \|A v_{T_i}\|_1  \]
	
	Since $v_{S \cup T_1}$ and $v_{T_i}$ are all at most $(1+\alpha^2)k$-sparse, 
	
	\[ L \|v_{S \cup T_1}\|_2 - \sum_{i > 1} \|A v_{T_i}\|_1 \leq \|Av\|_1 \leq U \|v_{S \cup T_1}\|_2 + \sum_{i > 1} \|A v_{T_i}\|_1  \]
	
	We now prove an upper bound on the quantity $\sum_{i > 1} \|Av_{T_i}\|_1$. This will give us both the upper and lower bounds we need. To this end, observe that all coordinates of $v_{T_{i-1}}$ are greater than or equal to all coordinates of $v_{T_i}$. This implies 
	\[ \|v_{T_{i}}\|_{\infty} \leq \frac{\|v_{T_{i-1}}\|_1}{\alpha^2 k} \]
	which, in turn, implies 
	\[ \|v_{T_i}\|_2 \leq \frac{1}{\alpha \sqrt{k}} \|v_{T_{i-1}}\|_1.\]
	Using the bounds on the restricted sparse eigenvalues from the statement, we get
	\begin{align*}
	\sum_{i > 1} \|Av_{T_i}\|_1  &\leq U \cdot \sum_{i > 1} \|v_{T_i}\|_2 \\
	&\leq   \frac{U}{\alpha \sqrt{k}} \|v_{\overline{S}}\|_1\\
	&\leq  \frac{U}{\alpha \sqrt{k}} \left(\|v_{S}\|_1 + \Delta\right)\\
	&\leq  \frac{U}{\alpha} \cdot \|v_{S}\|_2 +  \frac{U}{\alpha \sqrt{k}} \cdot \Delta\\
	&\leq \frac{U}{\alpha} \cdot \|v_{S \cup T_1}\|_2 +  \frac{U}{\alpha \sqrt{k}} \cdot \Delta
	\end{align*}	
	
	Using the inequality above in addition to the bounds on $\|Av\|_1$, we get after some rearrangement
	
	\[\left( L  -  \frac{U}{\alpha} \right) \|v_{S \cup T_1}\|_2 -  \frac{U}{\alpha \sqrt{k}} \cdot \Delta \leq \|Av\|_1 \leq U \left(1 + \frac{1}{\alpha}\right) \|v_{S \cup T_1}\|_2 +  \frac{U}{\alpha \sqrt{k}} \cdot \Delta \]
	
	The bounds above are in terms of $\|v_{S \cup T_1}\|_2$, however we need bounds in terms of $\|v\|_2$. For the upper bound, it is sufficient to note that $\|v_{S \cup T_1}\|_2 < \|v\|_2$. For the lower bound, we need the inequalities below. 
	
	The definition of $T_i$ and applications of the Cauchy-Schwartz inequality gives us
	\[ 
	\|v_{\overline {S \cup T_1}}\|_2 \leq \sum_{i \geq 2} \|v_{T_i} \|_2 \leq \frac{1}{\alpha \sqrt{k}} \sum_{i \geq 1}  \|v_{T_{i-1}}\|_1 \leq \frac{\|v_{\overline{S}}\|_1}{\alpha \sqrt{k}} \leq \frac{\|v_{S}\|_1 + \Delta}{\alpha \sqrt{k}} \leq \frac{\|v_{S}\|_2}{\alpha} + \frac{\Delta}{\alpha \sqrt{k}}.\]
	This, in turn, results in an upper bound on $\|v\|_2$ in terms of $\|v_{S \cup T_1}\|_2$,
	\begin{align*}
	\|v\|_2 &\leq \|v_{S \cup T_1}\|_2 + \|v_{\overline{S \cup T_1}}\|_2 \\
	&\leq  \|v_{S \cup T_1}\|_2 +  \sum_{i \geq 2} \|v_{T_i} \|_2 \\
	&\leq \|v_{S \cup T_1}\|_2 + \frac{\|v_{S}\|_2}{\alpha} + \frac{\Delta}{\alpha \sqrt{k}}\\
	&\leq \left(1+ \frac{1}{\alpha}\right) \|v_{S \cup T_1}\|_2 + \frac{\Delta}{\alpha \sqrt{k}}\\
	&\implies \|v_{S \cup T_1}\|_2 \geq  \frac{\alpha}{1+\alpha} \left(\|v\|_2 - \frac{\Delta}{\alpha \sqrt{k}}\right) 
	\end{align*}
	and so
	\[ \frac{L}{1+\alpha} \left( \alpha  -  \frac{U}{L} \right)  \left(\|v\|_2 - \frac{\Delta}{\alpha \sqrt{k}} \right) -  \frac{U \Delta}{\alpha \sqrt{k}} \leq \|Av\|_1 \leq U \left(1 + \frac{1}{\alpha}\right) \|v_{S \cup T_1}\|_2 +  \frac{U \Delta}{\alpha \sqrt{k}} \]
	At this point, we have the upper bound, to complete the proof of the lower bound, observe that standard manipulations give us
	\begin{align*}
	\frac{L}{1+\alpha} \left( \alpha  -  \frac{U}{L} \right)  \left(\|v\|_2 - \frac{\Delta}{\alpha \sqrt{k}} \right) &= \frac{L}{1+\alpha} \left( \alpha  -  \frac{U}{L} \right)\|v\|_2 - \frac{L}{1+\alpha} \left( \alpha  -  \frac{U}{L} \right) \frac{\Delta}{\alpha \sqrt{k}}\\
	&= \frac{L}{1+\alpha} \left( \alpha  -  \frac{U}{L} \right)\|v\|_2 - \frac{ \alpha L - U}{1+\alpha} \frac{\Delta}{ \alpha \sqrt{k}}\\
	&\geq \frac{L}{1+\alpha} \left( \alpha  -  \frac{U}{L} \right)\|v\|_2 -  \frac{ U \Delta}{ \alpha \sqrt{k}}
	\end{align*}
	
	This gives us the Lemma, 
	
	\[  \frac{L}{1+\alpha} \left( \alpha  -  \frac{U}{L} \right)\|v\|_2 -  \frac{2U\Delta}{ \alpha \sqrt{k}}  \leq \|Av\|_1 \leq U \left(1 + \frac{1}{\alpha}\right) \|v_{S \cup T_1}\|_2 +  \frac{U \Delta}{\alpha \sqrt{k}}. \]
	
\end{proof} 
%

\end{document}